\documentclass[journal,twoside,web]{ieeecolor}

\usepackage{generic}

\usepackage[utf8]{inputenc}
\usepackage{cite}
\usepackage[utf8]{inputenc}
\usepackage{cite}
\usepackage{qtree}
\usepackage{tikz}
\usepackage{amsmath,amssymb,amsfonts}
\usepackage{eufrak}
\usepackage{dsfont}
\usepackage{steinmetz}
\usepackage{algorithmicx}
\usepackage{algorithm}
\usepackage{algpseudocode}
\usepackage{graphicx}
\usepackage{amsmath,amssymb}
\usepackage{textcomp}
\usepackage{algpseudocode} 
\usepackage{mathtools}
\usepackage{nicefrac}
\usepackage{txfonts}
\usepackage{tikz}
\usepackage{dsfont}
\usepackage{graphicx}
\usepackage{amsmath,amssymb}
\usepackage{xcolor}
\usepackage{mathtools}
\usepackage{nicefrac}
\usepackage{txfonts}
\usepackage{todonotes}
\usepackage{cleveref}

\newcommand{\rev}[1]{{\color{black} #1}}

\title{Cooperative Target Defense under Communication and Sensing Constraints}

\author{ Dipankar Maity, Arman Pourghorban
\thanks{The authors are with the Department of Electrical and Computer Engineering, University of North Carolina at Charlotte, NC, 28223, USA. 
Email: 
{\tt dmaity@charlotte.edu, apourgho@charlotte.edu}
}%
\thanks{
This research is supported by the ARL grant ARL DCIST CRA W911NF-17-2-0181
}
}

\newcommand{\re}[1]{{\color{red} #1}}

\newcommand{\T}{^{\mbox{\tiny \sf T}}}

\newcommand{\xa}{\mathbf{x}_{N+1}}

\newcommand{\xii}{\boldsymbol{\xi}}

\newcommand{\uvec}{\hat{\mathbf{u}}}
\newcommand{\x}{\mathbf{x}}
\newcommand{\sign}{\text{sgn}}

\renewcommand{\Re}{\mathbb{R}}

\newtheorem{lemma}{Lemma}

\newtheorem{theorem}{Theorem}

\newtheorem{remark}{Remark}

\newtheorem{assm}{Assumption}

\begin{document}

\maketitle
\thispagestyle{empty}
\begin{abstract}
We consider a variant of the target defense problems where a group of defenders are tasked to simultaneously capture an intruder. The intruder's objective is to reach a target without being simultaneously captured by the defender team. 
Some of the defenders are sensing-limited and do not have any information regarding the intruder's position or velocity at any time. 
The defenders may communicate with each other using a connected communication graph. 
We propose a decentralized feedback strategy for the defenders, which transforms the simultaneous capture problem into a unique nonlinear consensus problem. 
We derive a sufficient condition for simultaneous capture in terms of the agents speeds, sensing, and communication capabilities. 
The proposed decentralized controller is evaluated through extensive numerical simulations.

\end{abstract}


\section{Introduction}

\IEEEPARstart{C}{oordinated} target-defense problems, where a group of defenders cooperatively guards a static or dynamic target against intruders, have attracted significant research interests \cite{shishika2020review, jajodia2012moving}. 
\textit{Simultaneous capture} is often a key requirement for this type of problems; \rev{see \cite{jeon2010homing,kumar2019cooperative,sinha2021cooperative} and the references therein}.
That is, all the defenders shall simultaneously make contact with the intruder, requiring precise coordination and cooperation among the defenders to agree on the location where they will rendezvous with the intruder\rev{ \cite{yan2020distributed,zhou2016distributed}}.

%
Existing approaches typically involve a game theoretic solution which assumes each and every agent knows others' initial locations, dynamics, objective function, and capabilities (e.g., maximum speed, turn radius etc.) \cite{chen2016multi}. 
These methods seek an equilibrium strategy for the defenders and the intruder. 
Feedback type equilibrium strategies are developed where each agent's instantaneous control action depends on the states of all the agents. 
While these game theoretic methods are very powerful, they unfortunately do not extend to scenarios where, for instance, the dynamics of the intruder is unknown, or some of the defenders cannot sense the intruder, or some defenders can only communicate  with a part of the defender team, \rev{e.g., in certain military applications \cite{lee2023graph}}.
Differential games with incomplete knowledge \cite{shishika2024deception}, limited sensing \cite{maity2023optimal}, and communication \cite{maity2023efficient} are extremely challenging to deal with, and only a little is known to date that can be applied in this challenging problem.
Therefore, these practical challenges pertaining to limitations in sensing and communication, and incomplete knowledge about the opponent render a game theoretic formulation to be inadequate and intractable. 

In this letter, we formulate this target-defense scenario as a decentralized control problem for the defenders. 
Our main objective is to develop defenders' control strategies for \textit{simultaneous capture} via enabling each defender to leverage its own sensing and communication capability to coordinate with the rest of the group and come to a \textit{dynamic consensus} on the capture location and time. 
%
%
To this end, we borrow ideas from \textit{multi-agent dynamic consensus} problems \cite{olfati2007consensus}---\rev{more precisely, the \textit{leader-follower} type formulations \cite{hong2006tracking}}---to obtain a solution with theoretical guarantees.
Our method does not require the knowledge of the intruder's dynamics, objective function, or its speed capability, making the method applicable to a wide range of problems.

\rev{A key distinction between our approach and much of the existing literature on (leader-follower) consensus problems is that our consensus algorithm must be specifically designed to satisfy the physical limitations of the defenders (e.g., maximum speeds). For instance, the consensus-based optimal control approach for multi-agent pursuit-evasion games \cite{lopez2019solutions} does not account for physical speed limitations (i.e., bounds on control magnitudes). Additionally, linear consensus protocols \cite{olfati2007consensus}  are unsuitable for two primary reasons: (i) they achieve asymptotic consensus, whereas our defenders must capture targets within finite time, and (ii) they do not guarantee that the resulting controllers respect the physical actuation limits of the defenders.}
%
%

To the best of our knowledge, this is the first work to analyze a sensing and communication limited simultaneous capture problem for an unknown intruder dynamics. 
We derive a sufficient condition on the sensing, communication, and speed capabilities of the defenders to ensure simultaneous capture. 
We then derive an upper bound on the capture time. 
Our analysis also provides both qualitative and quantitative understanding on the sensing versus communication tradeoff required for simultaneous capture.

\section{Problem Formulation} \label{sec:ProbFormulation}
We consider a target defense problem in ${\mathbb{R}^2}$, similar to that studied in \cite{bajaj2019dynamic, shishika2021partial, pourghorban2022target}, where an intruder attempts to breach a point target $\mathcal{T}$.\footnote{In cases where the target is a region, $\mathcal{T}$ may represent either the center of the target or a point on its perimeter.} 
A group of $N$ defenders is tasked to simultaneously capture the intruder before it reaches the target.
Not all the defenders are equipped with sensing capabilities and they do not have the location of the intruder. 
The defenders, however, can communicate with each other via a connected communication graph $\mathcal{G}$ to exchange relevant information for simultaneously capturing the intruder.

As considered in the existing target defense work \cite{fuchs2010cooperative, yan2021cooperative, liang2020analysis}, the dynamics of the defenders and the intruder are given by
\begin{align} \label{eq:dyn}
    \dot{\x}_i(t) = v_i \uvec(\theta_i(t)), \qquad \forall ~~i =1,\dots, (N+1),
\end{align}
where $\uvec(\theta) = [\cos\theta, \sin\theta]\T$ is the unit vector making an angle of $\theta$ with the $x$-axis. $\x_i(t)$ and $\x_{N+1}(t)$ denote the positions of the $i$-th defender and the intruder, respectively, at time $t$.
The control inputs for the defenders and the intruder are their instantaneous heading angle $\theta_i(t)$. 
Their speeds are constant and denoted by $v_i$.  

The requirement for simultaneous capture is crucial; without it, the defenders cannot neutralize the intruder.
For any $i=1,\ldots,N$, we say that the $i$-th defender has captured the intruder at time $t$ if and only if $\|\x_i(t) - \x_{N+1}(t)\| = 0$.
Simultaneous capture happens at time $t$ if and only if $\sum_{i=1}^N \| \x_i(t) -\x_{N+1}(t)\| = 0$.
The requirement for simultaneous capture in our problem distinguishes it from prior research \cite{bajaj2019dynamic, shishika2021partial, pourghorban2022target}, where a single defender can neutralize the intruder upon contact. 
In our scenario, not all defenders have the capability to sense the intruder. 
This key distinction makes the simultaneous capture problem significantly more interesting and challenging than existing research on simultaneous capture \cite{von2019multi, von2020robust} where all the defenders can sense the intruder.

To reflect the effect of capture on the intruder's dynamics, we consider a  more accurate version of \eqref{eq:dyn} for the intruder's dynamics:
\begin{align}\label{eq:intruder_dynamics}
\dot{\x}_{N+1}(t)= \delta(\x)  v_{N+1} \uvec(\theta_{N+1}(t)) ,
\end{align}
where $\x=[\x_1\T,\dots,\x_N\T,\x_{N+1}\T]\T$ denotes the aggregated states of all the agents, and 
\begin{equation} \label{eq:delta_x}
    \delta(\x) = \sign\left(\sum_{i=1}^N \| \x_i -\x_{N+1}\|\right) 
\end{equation}
is an indicator variable denoting whether the intruder has been captured or not. 
Here,
\begin{align*}
    \sign(x) = \begin{cases}
        \frac{x}{|x|}, \quad &x \ne 0,\\
        0, & x=0
    \end{cases}
\end{align*}
denotes the sign function.
%
Once the simultaneous capture happens, we will have $\dot{\x}_{N+1}(t) = 0$ afterwards, denoting that the intruder does not move any more. 
Dynamics \eqref{eq:dyn} without the $\delta(\x)$ term fail to model this aspect.

The intruder may have access to $\x(t)$ or only a part of it at any time $t$, which the intruder uses to construct its instantaneous heading angle $\theta_{N+1}(t)$. 
The policy for computing $\theta_{N+1}(t)$ depends on the intruder's objective function, which is assumed to be unknown to the defenders.\footnote{ Note that the objective of reaching the target can be expressed as an optimization problem with several different objective functions. 
Each objective function results in a particular strategy for computing $\theta_{N+1}(t)$.} 
One particular strategy is to head directly toward the target (i.e., shortest path) with maximum speed. 
In this letter, we assume that the defenders are unaware of the intruder's control policy and therefore cannot estimate $\theta_{N+1}(t)$. 
In other words, the defenders perform simultaneous capture for an unknown intruder model. 
%
The objective in this letter is to find sufficient condition on the defenders' speed, sensing, and communication capabilities for simultaneous capture and derive an upper bound on the capture time.

Without loss of generality, we may assume that the target is at the origin of our coordinate system, i.e., $\mathcal{T}=0$.
Otherwise, we may study the problem in a shifted coordinate system by considering the state $\Tilde{\x}_i \triangleq \x_i - \mathcal{T}$ instead of $\x_i$ for all $i$.

\section{Solution Approach: A Consensus Problem}

Although simultaneous capture problems \cite{von2019multi, von2020robust} have been studied in the past, our problem does not fit those existing formulations for two main reasons, both of them are related to \textit{information structures}:  
(i) Only a subset of the defenders can sense the intruder, which results in \textit{partial and asymmetric state information} for the defenders, and (ii) intruder's control objective and the speed parameter $v_{N+1}$ are not be known to the defenders, resulting in an \textit{incomplete information} game.

Our proposed solution in this letter is based on a consensus principle, where the defenders will continuously exchange their locations $\x_i(t)$ with each other to agree on a capture point in an adaptive fashion. 
The solution does not use/exploit the dynamics \eqref{eq:intruder_dynamics} and therefore, can be applicable for other simultaneous capture problems with unknown evader dynamics of the form $\dot{\x}_{N+1} = f(\x)$.

Consensus algorithms, as discussed in \cite{olfati2007consensus}, have proven to be a classical approach for addressing control theoretic problems with partial information structures.
In our formulation, where defenders can communicate with each other, a consensus approach naturally emerges as an efficient means to compute the defenders' control actions for simultaneous capture. 
While it is important to note that a consensus-based approach may not always yield an optimal solution for simultaneous capture, we demonstrate in this letter that it is an unique and powerful tool with significant potential for addressing this type of problems.
\subsection{Consensus Protocol for Simultaneous Capture}
\label{consensusdynamicssec}

%
We prescribe the defenders to follow the dynamics
\begin{align}
\label{nonlindynamics}
    \dot{\x}_i(t)=v_i \frac{\sum \nolimits _{{j=1}}^N[w_{i j}(\x_j(t)-\x_i(t))]+b_i(\x_{N+1}(t) - \x_i(t))}{\| \sum \nolimits _{{j=1}}^N[w_{i j}(\x_j(t)-\x_i(t))]+b_i(\x_{N+1}(t) - \x_i(t)) \|}, 
\end{align}
where $w_{ij} \ge 0$ for all $i, j$, and  $w_{ij} > 0 $ if and only if defender $j$ can send messages to defender $i$. 
If the $i$-th defender can sense the intruder, then $b_i = 1$ otherwise $b_i = 0$. 
The variables {\small $\{w_{ij}\}_{j = 1}^N$} denote the communication capability of defender $i$, and $b_i$ denote its sensing capability. 

One may verify that the proposed defender dynamics \eqref{nonlindynamics} is indeed of the form \eqref{eq:dyn}, with an appropriate choice for $\theta_i(t)$. 
The proposed nonlinear consensus dynamics ensure \textit{finite time consensus}, as will be shown in \Cref{thm:main}.

Let us define the matrix $W \in \Re^{N \times N}$ such that
\begin{align} \label{eq:W}
    [W]_{ij} = \begin{cases}
        -w_{ij}, \qquad &i\ne j \\ 
        \sum_{j=1}^{N}w_{ij} + b_i, \qquad &i = j.  
    \end{cases}
\end{align}
We will make the following assumptions in this work.
\begin{assm} \label{assm:communication}
    The communication graph is symmetric (i.e., $w_{ij}= w_{ji}$ for all $i,j$) and connected. 
\end{assm}
The symmetry assumption simplifies the analysis; however, without a strongly connected communication graph, simultaneous capture is not guaranteed. Such graphs are generally required for consensus algorithms.
\begin{assm} \label{assm:sensing}
    There exists an $i$, such that $b_i = 1$.
\end{assm}
This assumption is to ensure that at least one defender can sense the intruder.

To proceed with our analysis, let us define the variables $\xii_i(t) = \x_i(t) - \xa(t)$. 
Using \eqref{nonlindynamics}, we obtain the dynamics of $\xii_i$'s as follows:
\begin{align}
    \label{eq:xi_dynamics}
 \begin{split}   
    \dot{\xii}_i(t) = & v_i \frac{\sum \nolimits _{{j=1}}^N[w_{i j}(\xii_j(t)-\xii_i(t))]- b_i \xii_i(t)}{\| \sum \nolimits _{{j=1}}^N[w_{i j}(\xii_j(t)-\xii_i(t))] - b_i \xii_i(t)) \|} \\
     & \qquad - v_{N+1}\begin{bmatrix}
        \cos(\theta_{N+1}(t)) \\
        \sin(\theta_{N+1}(t))
    \end{bmatrix}.
    \end{split}
\end{align}
We may now rewrite dynamics \eqref{eq:xi_dynamics} as 
\begin{align}
\label{nonlindynmat}
    \dot{\xii}_i(t)= - v_i\frac{([W]_i \otimes I_2) \xii(t)}{\| ([W]_i \otimes I_2) \xii(t)\|} - v_{N+1} \begin{bmatrix}
        \cos(\theta_{N+1}(t)) \\
        \sin(\theta_{N+1}(t))
    \end{bmatrix},
\end{align}
where $[W]_i$ denotes the $i$-th row of matrix $W$, $I_2$ is the two dimensional identity matrix, and we have defined $\xii=[\xii_1\T,\dots,\xii_N\T]\T$.
By further defining 
\begin{align} \label{eq:Fx}
\begin{split}
        F(\xii)        =\textrm{diag}\bigg\{\frac{v_1}{\|{([W]_1}\otimes I_2)\xii\|}, &\ \frac{v_2}{\|{([W]_2} \otimes I_2)\xii\|}, \  \dots, \\
        &\  \frac{v_{N }}{\|{([W]_{N} \otimes I_2)}\xii\|} \bigg\},
\end{split}
\end{align}
and $\bar W(\xii) = (F(\xii) W) \otimes I_2$, we may compactly write \eqref{nonlindynmat} as
\begin{align} \label{eq:compact_dynamics}
    \dot\xii = - \bar{W}(\xii) \xii - v_{N+1}(\mathds{1}\otimes I_2)    \begin{bmatrix}
        \cos(\theta_{N+1}(t)) \\
        \sin(\theta_{N+1}(t))
    \end{bmatrix},
\end{align}
where $\mathds{1} \in \Re^{N}$ is a vector of all ones.
The following lemma will be instrumental in our subsequent consensus analysis. 
For this lemma, let us write 
$W$ in \eqref{eq:W} as follows:
\begin{align} \label{eq:W_decomposition}
    W = \underset{W_1}{\underbrace{\begin{bmatrix}
        ~w_{11} & \cdots & -w_{1N}\\
        \vdots & \vdots & \vdots\\
        -w_{N1} & \cdots & ~w_{NN}\\
    \end{bmatrix}}} + 
    \underset{W_2}{\underbrace{\begin{bmatrix}
        b_1 & \cdots & 0\\
        \vdots & \vdots & \vdots\\
        0 & \cdots & b_N\\
    \end{bmatrix}}},
\end{align}
which decomposes the `communication' and `sensing' capabilities into two separate matrices $W_1$ and $W_2$, respectively. 
Due to \Cref{assm:sensing}, we have $W_2 \ne 0$.
\begin{lemma}
    \label{lem:W_properties}
    For the given $W$ matrix in \eqref{eq:W_decomposition}, we have
    \begin{enumerate}
        \item $W$ is a positive-definite matrix. 
        \item $ \lambda_{\min}(W) \ge \min_{\gamma} \frac{\lambda_2(W_1) + \frac{m}{N}\left(|\gamma| - \sqrt{\frac{N-m}{m}} \right)^2}{\gamma^2 + 1}$,
        where $\lambda_{\min}(\cdot)$ and $\lambda_2(\cdot)$ denote the smallest and second smallest eigenvalue of a matrix, respectively, and $m = \sum_{i=1}^N b_i$.
    \end{enumerate}
\end{lemma}
\begin{proof}
    The proof is presented in Appendix~\ref{AP:W_properties}.
\end{proof}
One may solve the minimization over $\gamma$ in \Cref{lem:W_properties} quite straightforwardly.
However, the ultimate result from such minimization does not provide much fundamental insights on how $m/N$ and $\lambda_2(W_1)$ may affect the minimum eigenvalue of $W$. 
Instead, some qualitative understanding on the relationship between $\lambda_{\min}(W)$ and $\lambda_2(W_1)$ and $m/N$ will be more useful.
To that end, it is noteworthy that the lower bound on $\lambda_{\min}(W)$ increases with $\lambda_2(W_1)$, which has significant importance,
as we will later show how this could help control the simultaneous capture time or find sufficient conditions on agents' velocities for simultaneous capture
(c.f.~\Cref{sec:tradeoff}.)

\section{Consensus Analysis}
In this section we analyze the consensus dynamics and derive an upper bound on the consensus time. 
The main result of this work is presented next in \Cref{thm:main}. 
\begin{theorem} \label{thm:main}
    Let $v_{\min} \sqrt{\lambda_{\min}(W)}  - v_{N+1}\sqrt{\sum_{i=1}^N b_i} > 0$. Then simultaneous capture happens at $t^*$, where 
\begin{align}
\label{approximatetime}
    t^*\leq  \frac{\sqrt{\xii(0)\T (W\otimes I_2) \xii(0)}}{v_{\min} \sqrt{\lambda_{\min}(W)}  - v_{N+1}\sqrt{\sum_{i=1}^N b_i}}.
\end{align} \vspace{.25 cm}
\end{theorem}

\begin{proof}
    Consider the function $V= \xii\T (W\otimes I_2)\xii$, where $W$ is defined in \eqref{eq:W}. 
    Due to the positive definiteness of $W$ (c.f.~Lemma~\ref{lem:W_properties}), we have $V > 0$ for all $\xii \ne 0$. 
    Furthermore, $V(t) = 0$ ensures $\xii(t) = 0$, i.e., consensus occurs.  

Taking the time derivative of $V$ and using \eqref{eq:compact_dynamics}, we may write 
\begin{align*} 
    \dot{V} & = - \xii\T \big((W\otimes I_2) \bar W(\xii) +\bar W\T(\xii) (W\otimes I_2)\big)\xii \\
     &  -2 v_{N+1}\xii\T (W\otimes I_2)  (\mathds{1}\otimes I_2) \uvec(\theta_{N+1}) 
     \nonumber \\
     & = -2 \xii\T((WF(\xii)W) \otimes I_2)\xii - 2 v_{N+1}\xii\T (W\otimes I_2)  (\mathds{1}\otimes I_2) \uvec(\theta_{N+1})
\end{align*}
where we have used $(A\otimes B)(C \otimes B) = (AC)\otimes B$ and $(A\otimes B)\T = A\T \otimes B\T$. 
Given that $F(\xii)$ is diagonal, we may further simplify $\xii\T((WF(\xii)W) \otimes I_2)\xii$ and express $\dot{V}$ as
\begin{align*}
    \dot{V} 
     & = -2 \sum\nolimits_{i=1}^N v_i \|([W]_i \otimes I_2) \xii \| - 2 v_{N+1}\xii\T (W\otimes I_2)  (\mathds{1}\otimes I_2) \uvec(\theta_{N+1}),
\end{align*}
where we have used \eqref{eq:Fx} to substitute the expressions for {\small $[F(\xii)]_{ii}$}.
To upper bound $\dot{V}$, we use the Cauchy-Schwartz inequality:
\begin{align*}
    \xii\T (W\otimes I_2)  & (\mathds{1}\otimes I_2) \uvec(\theta_{N+1}) \\
    &\le \|(W^{\frac{1}{2}}\otimes I_2) \xii\| \|(W^{\frac{1}{2}}\otimes I_2)(\mathds{1}\otimes I_2) \uvec(\theta_{N+1})\|
\end{align*}
where $W^{\frac{1}{2}}$ is a matrix such that $(W^{\frac{1}{2}})\T W^{\frac{1}{2}} = W$. 
Now, notice that $\|(W^{\frac{1}{2}}\otimes I_2) \xii\| = \sqrt{V}$, and we may simplify 
$\|(W^{\frac{1}{2}}\otimes I_2)(\mathds{1}\otimes I_2) \uvec(\theta_{N+1})\|$ as 
\begin{align*}
    \|(W^{\frac{1}{2}}\otimes I_2)(\mathds{1}\otimes I_2) \uvec(\theta_{N+1})\|^2 &= (\mathds{1}\T W \mathds{1}) \uvec(\theta_{N+1})\T \uvec(\theta_{N+1}) \\
    & = \mathds{1}\T W \mathds{1} = \sum\nolimits_{i=1}^N b_i.
\end{align*}
Therefore, we may now upper bound $\dot{V}$ as
\begin{align} \label{eq:partial_upper_bound_V_dot}
    \dot{V} \le - 2 \sum\nolimits_{i=1}^N v_i \|([W]_i \otimes I_2) \xii \| + 2 v_{N+1} \sqrt{\sum\nolimits_{i=1}^N b_i} \sqrt{V}.
\end{align}
Using similar arguments, we may verify that 
\begin{align*}
    \sum\nolimits_{i=1}^N v_i \|([W]_i \otimes I_2) \xii \| & \ge v_{\min} \|(W \otimes I_2)\xii\|\\
    \ge v_{\min} \sqrt{\lambda_{\min}(W)} & \|(W^{\frac{1}{2}} \otimes I_2)\xii\| = v_{\min} \sqrt{\lambda_{\min}(W)} \sqrt{V}.
\end{align*}
Thus, we may write from \eqref{eq:partial_upper_bound_V_dot}
\begin{align}
    \dot{V} \le -2c\sqrt{V},
\end{align}
where $c = v_{\min} \sqrt{\lambda_{\min}(W)}  - v_{N+1}\sqrt{\sum_{i=1}^N b_i} $. 
Consequently, 
\begin{align*}
    \sqrt{V(t)} \le \sqrt{V(0)} - ct,
\end{align*}
for all $t\ge 0$. 
If $t^*$ is the time when consensus happens, then $V(t^*) = 0$, and hence
\begin{align*}
    t^* \le  \frac{\sqrt{\xii(0)\T (W\otimes I_2) \xii(0)}}{v_{\min} \sqrt{\lambda_{\min}(W)}  - v_{N+1}\sqrt{\sum_{i=1}^N b_i}} . 
\end{align*}
\end{proof}
\begin{remark}
The upper bound on the consensus time depends on the initial locations of all the agents, the matrix $W$, and the speeds $v_i$'s. 
The defender with the slowest speed greatly influences the consensus time. 
\end{remark}

\begin{remark}
    A sufficient condition for simultaneously capturing the intruder before it reaches the target is
    \begin{align} \label{eq:capture_sufficient_condition}
        \frac{\sqrt{\xii(0)\T (W\otimes I_2) \xii(0)}}{\frac{\|\x_{N+1}(0)\|}{\nu_{N+1}}  } \le v_{\min} \sqrt{\lambda_{\min}(W)}  - v_{N+1}\sqrt{\sum\nolimits_{i=1}^N b_i} ,
    \end{align}
    where $\frac{\|\x_{N+1}(0)\|}{\nu_{N+1}}$ is the earliest time the intruder may reach the target by heading directly to the target. 
    Simplifying \eqref{eq:capture_sufficient_condition} yields
    \begin{align} \label{eq:new_sufficient_condition}
        \frac{\nu_{\min}}{\nu_{N+1}} \ge \frac{1}{\sqrt{\lambda_{\min}(W)}} \left( \frac{\sqrt{\xii(0)\T (W\otimes I_2) \xii(0)}}{\|\x_{N+1}(0)\|} + \sqrt{\sum\nolimits_{i=1}^N b_i}~\right).
    \end{align}
\end{remark}

\subsection{Sensing versus Communication} \label{sec:tradeoff}
Combining \eqref{eq:new_sufficient_condition} and \Cref{lem:W_properties}, it suffices to have 
\begin{align} \label{eq:speed_conditions}
        v_{\min}^2 \min_{\gamma}  \frac{\frac{\lambda_2(W_1)}{m} + \frac{1}{N}\left(|\gamma| - \sqrt{\frac{N-m}{m}} \right)^2}{\gamma^2 + 1} \ge v_{N+1}^2  \alpha^2, 
    \end{align}
    where $\alpha =  \left( \frac{\sqrt{\xii(0)\T (W\otimes I_2) \xii(0)}}{\|\x_{N+1}(0)\|} + \sqrt{\sum\nolimits_{i=1}^N b_i}~\right)$. 
    Notice that, for a fixed $m$, the defenders can improve their performance (i.e., faster capture time) by maximizing the ratio  $\frac{\lambda_2(W_1)}{m}$, which denotes the relative strength in communication versus sensing. 
    While sensing is a requirement (i.e., satisfaction of \Cref{assm:sensing}), the strength of inter-agent communication (i.e., $\lambda_2(W_1)$) should be high, according to \eqref{eq:speed_conditions}. 

    Note that the sensing capability $m$ appears in $\alpha$.
    A higher sensing capability increases the r.h.s. in \eqref{eq:speed_conditions}, and  requires a higher value for $\lambda_2(W_1)$ to satisfy the condition. 
    This shows that one may not tradeoff communication capability for sensing. 
    Rather, when $m$ increases, qualitatively, $\lambda_2(W_1)$ should also increase. 
    Basically, the defenders should be more cohesive with each other than with the intruder.


As the problem is of simultaneous capture, the primary objective is the `simultaneous' part, as `capture' would not happen if any of the defenders are missing at the capture time.
Thus, in a way, the consensus among the defenders overshadows the `capture' part, which one may notice from \eqref{eq:speed_conditions} that a lower inter-agent connectivity i.e., smaller $\lambda_2(W_1)$ requires a longer capture time. 
Furthermore, when the defender communication graph is not connected, i.e., $\lambda_2(W_1) = 0$, then \eqref{eq:speed_conditions} cannot be satisfied.

While this letter has provided the sufficient condition for simultaneous capture,  the necessary condition may also provide further insights on the sensing versus communication strength tradeoff. 
We leave this as a potential future direction.

\section{Simulation Results}
Our simulation scenario involves four defenders $(N=4)$. 
We simulate the agents' trajectories, capture time, as well as we identify the initial conditions that lead to (or, do not lead to) a simultaneous capture. 
We also experiment with the speeds, sensing, and communication capabilities. 
In all of the following experiments $w_{ij} \in \{0,1\}$.
\subsection{Agents Trajectories and Consensus Time}
Here, the initial location of the attacker is $\x_{5}(0)=[-5,10]\T$ and its speed is $\nu_{5}=0.1$. We experimented two scenarios here: homogeneous and heterogeneous configurations.

\subsubsection{Homogeneous Defender Configuration}
Defenders' initial locations are symmetric with respect to the target with $\x_1(0)=[5,5]\T $, $ \x_2(0)=[-5,-5]\T $, $\x_3(0)=[-5,5]\T$, and $\x_4(0)=[5,-5]\T$. 
All of them have a speed of 1. 
Their communication graph is \textit{complete}, and all of them can sense the attacker.

\subsubsection{Heterogeneous Defender Configuration}
In this case, we kept the defenders' locations the same as before, but gave them heterogeneous speed, sensing, and communication capabilities. 
Here, $\nu_1=1.3, \nu_2=1.4, \nu_3=1.5, \nu_4=1.4$. 
Defenders 2 and 3 can not sense the attacker. Defenders 1 and 2 do not communicate with each other.

All the parameters for the homogeneous and heterogeneous configurations are chosen such that they satisfy the sufficient condition derived earlier. 
The agent trajectories are shown in Fig.~\ref{fig:capturetraj}: left one is for the homogeneous case, and the one on the right is for the heterogeneous case. 
The capture times for the homogeneous and heterogeneous scenarios are $14.01$ and $11.36$, respectively. 
Faster velocities for the defenders reduce the capture time.
The upper bounds on the capture times computed from \Cref{thm:main} are $48.41$ and $46.93$, respectively. 

\begin{figure}[t]
     \centering
     \includegraphics[trim= 70 60 70 80, clip, height = 4.0 cm, width=0.4\linewidth]{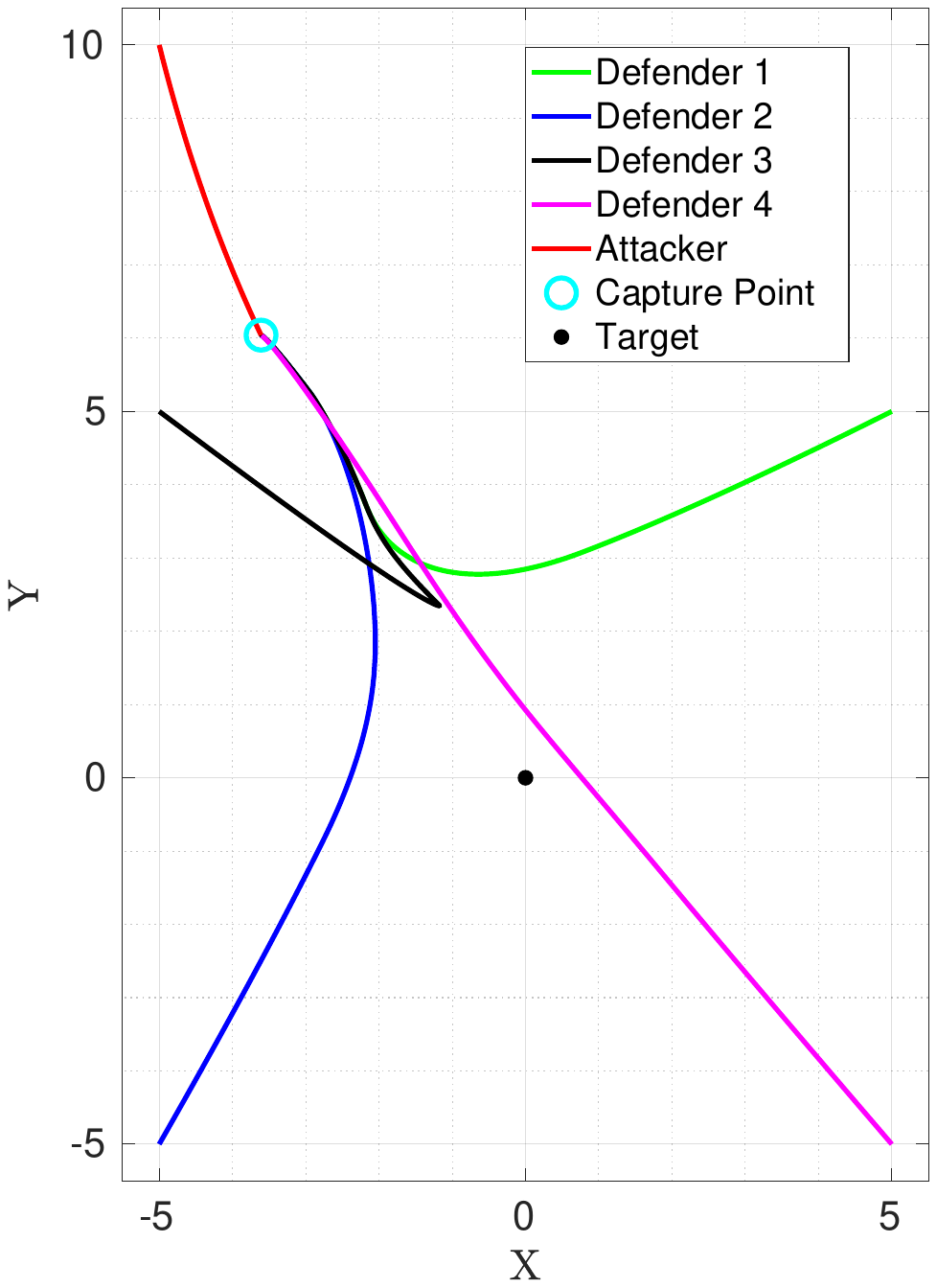} \includegraphics[trim= 70 60 70 80, clip, height = 4.0 cm, width=0.4\linewidth]{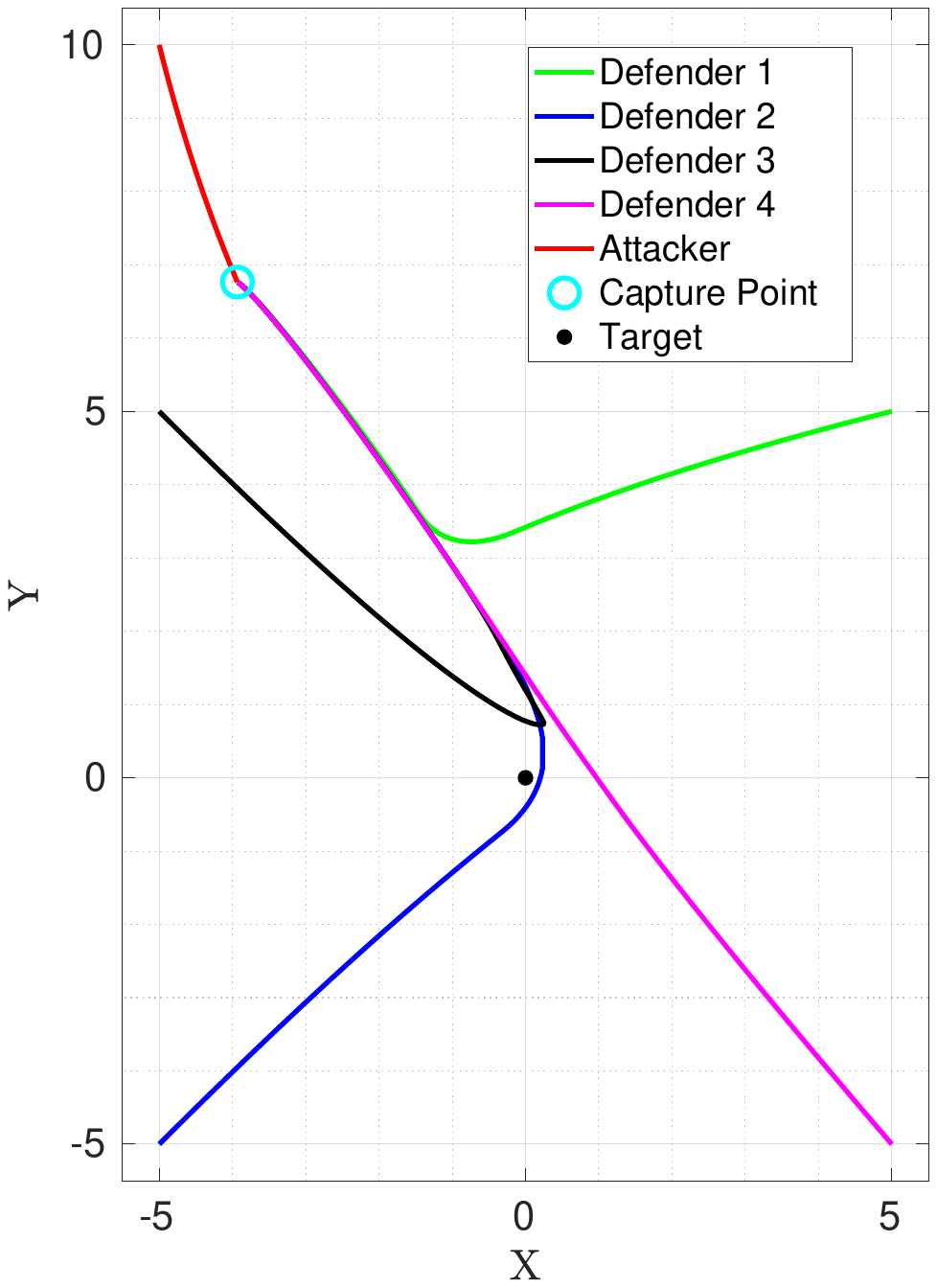}
     \caption{Left: Agents' trajectories in $\mathbb{R}^2$ where the defenders are homogeneous: $\nu_i=1$ for $i=\{1,2,3,4\}$, their communication graph is complete, and they all can sense the attacker. The intruder velocity is $\nu_5=0.1$. Right: Agents' trajectories in $\mathbb{R}^2$ for a heterogeneous group of defenders. Defender velocities are $\nu_1=1.3, \nu_2=1.4, \nu_3=1.5, \nu_4=1.4$. Defenders 2 and 3 do not sense the attacker and defenders 1 and 2 do not communicate with each other. The intruder velocity is still $\nu_5=0.1$.} \vspace{-1 em}
 \label{fig:capturetraj}
\end{figure}

\subsection{Capturable Initial Conditions}
Here, for a given set of parameters, the objective is to find the initial locations for the intruder so that simultaneous capture is guaranteed.
An illustration of that is presented in Fig.~\ref{fig:capturemap}, where the colorbar denotes the time to capture. 
Any intruder starting from the white region in the middle will be able to breach the target. 

The intruder speed is $0.5$. 
For the defenders, we consider the homogeneous parametric setting: $\x_1(0)=[5,5]\T $, $  \x_2(0)=[-5,-5]\T$ ,  $\x_3(0)=[-5,5]\T$ ,   $\x_4(0)=[5,-5]\T$.
All the defenders have a speed of $1$, and they all can sense the intruder, and their communication graph is complete. 
The capture time is symmetric with respect to the target, consistent with the homogeneous setup. As predicted in \eqref{approximatetime}, it increases with the attacker's distance from the target, reflecting the{\small$\sqrt{\xii(0)\T (W\otimes I_2) \xii(0)}$} term in \eqref{approximatetime}.
\begin{figure}
     \centering
     \includegraphics[trim= 100 245 140 252, clip, width=0.6\linewidth]{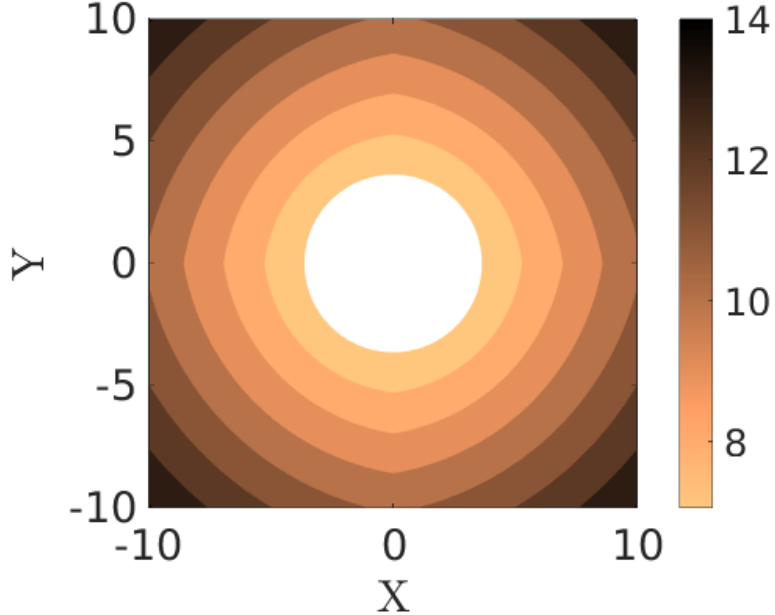} 
     \caption{Capture time for different starting points of the attacker and fixed defender locations. The white region represents the attacker's starting points from which a breach of the target is inevitable. 
     For other starting points of the attacker, the capture time is displayed using the colormap.}
 \label{fig:capturemap}
\end{figure}

\subsection{Parameter Variation Study}
The objective is to study how the boundary of the white region in Fig.~\ref{fig:capturemap} (\textit{non-capturable initial conditions}) changes with variations in speed, sensing, and communication parameters.

\subsubsection{Speed Variations}
We again consider a symmetric starting locations for the defenders: $\x_1(0)=[5,5]\T,~ \x_2(0)=[-5,-5]\T,~ \x_3(0)=[-5,5]\T, ~\x_4(0)=[5,-5]\T$, but vary the speed of fouth defender.  
The first three defenders have a speed of $1$ and the intruder has a speed of $0.1$. 
We consider five different values for $\nu_4 = \{0.2,0.4,0.6,0.8,1\}$.  
The capture region enlarges as the velocity of the defender increases, as one would expect. 
\begin{figure}
     \centering
     \includegraphics[trim= 37 110 30 120, clip, width=0.45\linewidth]{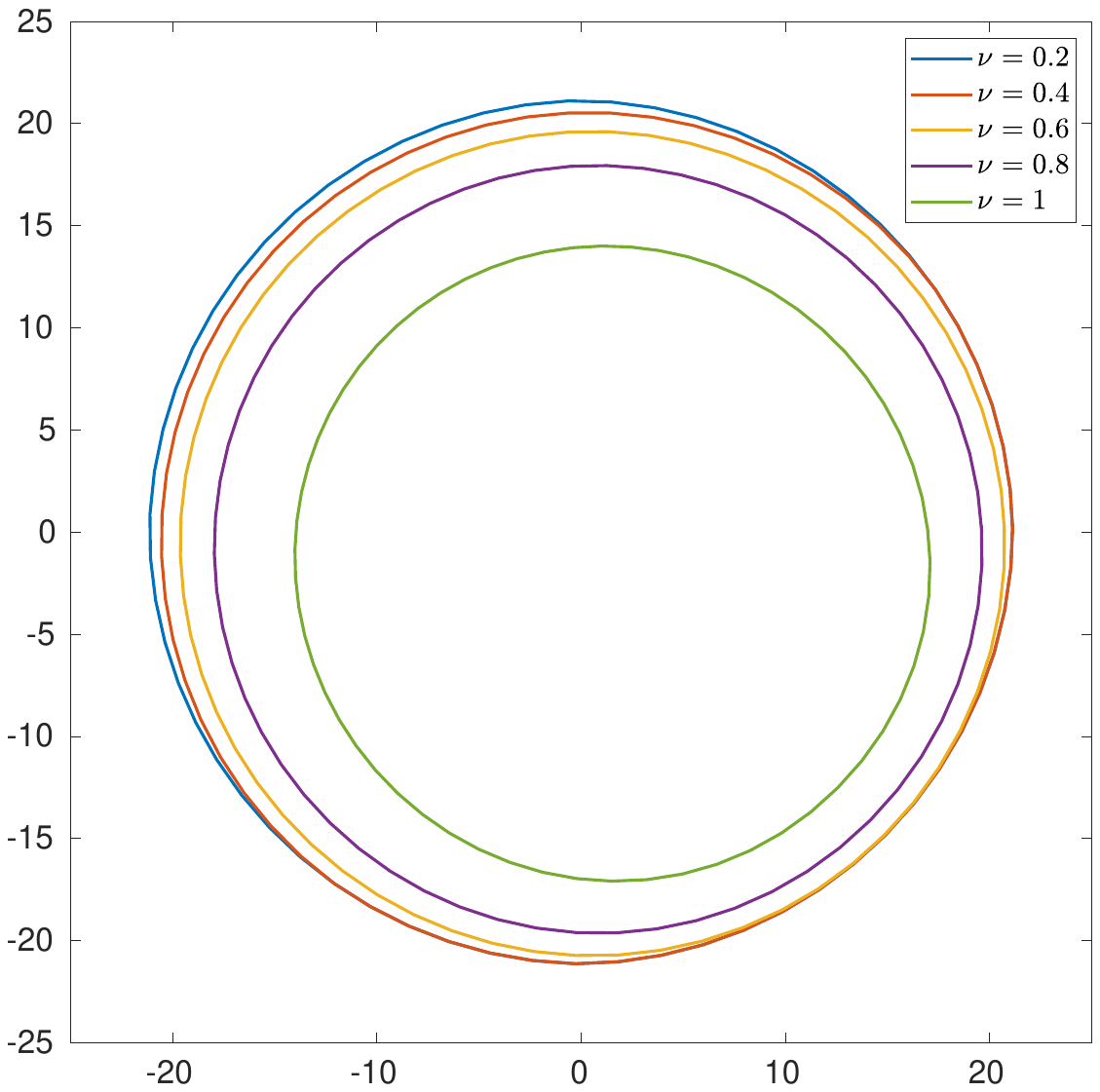} 
     \caption{Boundary of the non-capturable region as $\nu_4$ is varied. Here, $\nu_i=1$ for $i=\{1,2,3\}$ and $\nu_5=0.1$.} \vspace{-1 em}
 \label{fig:capturemap_velocity}
\end{figure}

\subsubsection{Sensing and Communication Variations} 
Here, we take $\x_1(0)=[2,10]\T,~ \x_2(0)=[-5,-6]\T,~ \x_3(0)=[-7,1]\T, ~\x_4(0)=[3,-1]\T$. 
All the defenders have speed of $1$, while the intruder has a speed of $0.6$. 
The sensing capability is represented as a `connectivity' with the intruder. 
Therefore, sensing and communication can be  jointly represented by a connectivity graph among the five agents, see Figs.~\ref{fig:capturemap_communication} and \ref{fig:capturemap_sensing}, where the intruder node is denoted by stars and the other four nodes denote the four defenders.

Fig.~\ref{fig:capturemap_communication} shows the effects of the agent communication graph on the \textit{non-capturable region}, and Fig.~\ref{fig:capturemap_sensing} illustrates the influence of sensing capability  on the \textit{non-capturable region}. 
The communication graph has more influence than the sensing capability, as discussed in \Cref{sec:tradeoff}.


\begin{figure}
     \centering
     \includegraphics[clip, width=0.6\linewidth]{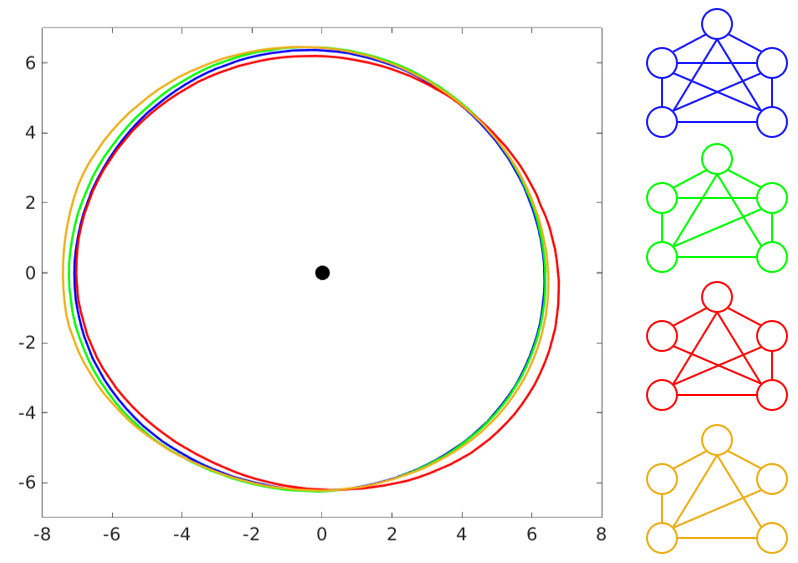}
     \put (-17.9,99.5) {{\color{blue}$ \ast$}}
     \put (-17.9,73.8) {{\color{green} $ \ast$}}
     \put (-17.8,21) {{\color{orange!60!yellow}$ \ast$}}
     \put (-17.8,47.9) {{\color{red} $ \ast$}}
     \caption{For different communication graphs, the boundary of the non-capturable region is shown. 
     Each region is shown in the same color as its corresponding communication graph. The black dot is the target $\mathcal{T}$.}
 \label{fig:capturemap_communication}
\end{figure}
 \begin{figure}
      \centering
      \includegraphics[ width=0.6\linewidth]{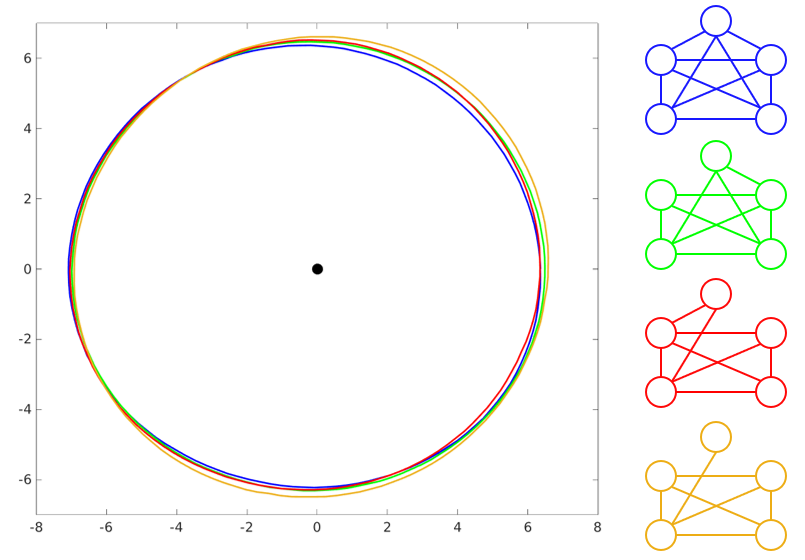} 
       \put (-19.1,99.2) {{\color{blue} $\ast$}}
     \put (-19.1,73.5) {{\color{green}$ \ast$}}
     \put (-19,47.6){$\re \ast$}
     \put (-19.1,20.6) {{\color{orange!60!yellow} $ \ast$}}
      \caption{For different sensing capabilities, the boundary of the non-capturable region is shown. 
     Each region is shown in the same color as its corresponding communication graph. The black dot is the target $\mathcal{T}$.}
  \label{fig:capturemap_sensing}
 \end{figure}
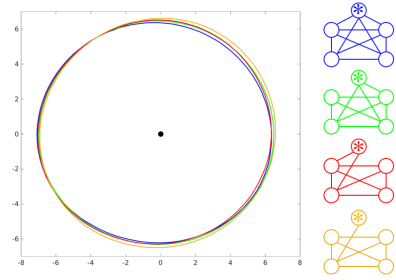

 \section{Conclusion}
\label{sec:Conclusion}
In this letter, we presented a coordinated target-defense game for a group of defenders against a single attacker. The defender team is required to capture the attacker simultaneously in order to win the game and the attacker is tasked to breach the target. Not all the defenders can sense the attacker and measure its location. However, the defenders can communicate with each other via an underlying
communication graph to simultaneously capture the attacker. We derive a sufficient condition on the defenders' sensing, communication, and speed capabilities to guarantee simultaneous capture. Using our derived expressions, we discuss the sensing versus
communication tradeoff required for simultaneous capture and provide an upper bound on the capture time. 
A natural extension of this work would be to consider a moving target, where the target is a part of the defender team.
One may also consider the same consensus problem under an event-triggered communication framework for reduction in inter-agent communication.
\bibliographystyle{IEEEtran}
\bibliography{main}

\begin{thebibliography}{10}
\providecommand{\url}[1]{#1}
\csname url@samestyle\endcsname
\providecommand{\newblock}{\relax}
\providecommand{\bibinfo}[2]{#2}
\providecommand{\BIBentrySTDinterwordspacing}{\spaceskip=0pt\relax}
\providecommand{\BIBentryALTinterwordstretchfactor}{4}
\providecommand{\BIBentryALTinterwordspacing}{\spaceskip=\fontdimen2\font plus
\BIBentryALTinterwordstretchfactor\fontdimen3\font minus \fontdimen4\font\relax}
\providecommand{\BIBforeignlanguage}[2]{{%
\expandafter\ifx\csname l@#1\endcsname\relax
\typeout{** WARNING: IEEEtran.bst: No hyphenation pattern has been}%
\typeout{** loaded for the language `#1'. Using the pattern for}%
\typeout{** the default language instead.}%
\else
\language=\csname l@#1\endcsname
\fi
#2}}
\providecommand{\BIBdecl}{\relax}
\BIBdecl

\bibitem{shishika2020review}
D.~Shishika and V.~Kumar, ``A review of multi agent perimeter defense games,'' in \emph{Decision and Game Theory for Security: 11th International Conference, GameSec 2020, College Park, MD, USA, October 28--30, 2020, Proceedings 11}.\hskip 1em plus 0.5em minus 0.4em\relax Springer, 2020, pp. 472--485.

\bibitem{jajodia2012moving}
S.~Jajodia, A.~K. Ghosh, V.~Subrahmanian, V.~Swarup, C.~Wang, and X.~S. Wang, \emph{Moving Target Defense II: Application of Game Theory and Adversarial Modeling}.\hskip 1em plus 0.5em minus 0.4em\relax Springer Science \& Business Media, 2012, vol. 100.

\bibitem{jeon2010homing}
\rev{Jeon, In-Soo and Lee, Jin-Ik and Tahk, Min-Jea}, ``\rev{Homing guidance law for cooperative attack of multiple missiles},'' \emph{Journal of guidance, control, and dynamics}, vol.~33, no.~1, pp. 275--280, 2010.

\bibitem{kumar2019cooperative}
\rev{Kumar, Shashi Ranjan and Mukherjee, Dwaipayan}, ``\rev{Cooperative salvo guidance using finite-time consensus over directed cycles},'' \emph{IEEE Transactions on Aerospace and Electronic Systems}, vol.~56, no.~2, pp. 1504--1514, 2019.

\bibitem{sinha2021cooperative}
\rev{Sinha, Abhinav and Kumar, Shashi Ranjan}, ``\rev{Cooperative Target Capture Using Predefined-Time Consensus},'' in \emph{2021 Seventh Indian Control Conference (ICC)}.\hskip 1em plus 0.5em minus 0.4em\relax IEEE, 2021, pp. 389--394.

\bibitem{yan2020distributed}
\rev{Yan, Pengpeng and Fan, Yonghua and Liu, Ruifan and Wang, Mingang}, ``\rev{Distributed target-encirclement guidance law for cooperative attack of multiple missiles},'' \emph{International Journal of Advanced Robotic Systems}, vol.~17, no.~3, p. 1729881420929140, 2020.

\bibitem{zhou2016distributed}
\rev{Zhou, Jialing and Yang, Jianying}, ``\rev{Distributed guidance law design for cooperative simultaneous attacks with multiple missiles},'' \emph{Journal of Guidance, Control, and Dynamics}, vol.~39, no.~10, pp. 2439--2447, 2016.

\bibitem{chen2016multi}
J.~Chen, W.~Zha, Z.~Peng, and D.~Gu, ``Multi-player pursuit--evasion games with one superior evader,'' \emph{Automatica}, vol.~71, pp. 24--32, 2016.

\bibitem{lee2023graph}
\rev{Lee, Elijah S and Zhou, Lifeng and Ribeiro, Alejandro and Kumar, Vijay}, ``\rev{Graph neural networks for decentralized multi-agent perimeter defense},'' \emph{Frontiers in Control Engineering}, vol.~4, p. 1104745, 2023.

\bibitem{shishika2024deception}
D.~Shishika, A.~Von~Moll, D.~Maity, and M.~Dorothy, ``Deception in differential games: Information limiting strategy to induce dilemma,'' \emph{arXiv preprint arXiv:2405.07465}, 2024.

\bibitem{maity2023optimal}
D.~Maity, ``Optimal intermittent sensing for pursuit-evasion games,'' \emph{IEEE Control Systems Letters}, 2023.

\bibitem{maity2023efficient}
------, ``Efficient communication for pursuit-evasion games with asymmetric information,'' in \emph{2023 62nd IEEE Conference on Decision and Control (CDC)}.\hskip 1em plus 0.5em minus 0.4em\relax IEEE, 2023, pp. 2104--2109.

\bibitem{olfati2007consensus}
R.~Olfati-Saber, J.~A. Fax, and R.~M. Murray, ``Consensus and cooperation in networked multi-agent systems,'' \emph{Proceedings of the IEEE}, vol.~95, no.~1, pp. 215--233, 2007.

\bibitem{hong2006tracking}
\rev{Hong, Yiguang and Hu, Jiangping and Gao, Linxin}, ``\rev{Tracking control for multi-agent consensus with an active leader and variable topology},'' \emph{Automatica}, vol.~42, no.~7, pp. 1177--1182, 2006.

\bibitem{lopez2019solutions}
V.~G. Lopez, F.~L. Lewis, Y.~Wan, E.~N. Sanchez, and L.~Fan, ``Solutions for multi-agent pursuit-evasion games on communication graphs: Finite-time capture and asymptotic behaviors,'' \emph{IEEE Transactions on Automatic Control}, vol.~65, no.~5, pp. 1911--1923, 2019.

\bibitem{bajaj2019dynamic}
S.~Bajaj and S.~D. Bopardikar, ``Dynamic boundary guarding against radially incoming targets,'' in \emph{58th IEEE Conference on Decision and Control}, 2019, pp. 4804--4809.

\bibitem{shishika2021partial}
D.~Shishika, D.~Maity, and M.~Dorothy, ``Partial information target defense game,'' in \emph{International Conference on Robotics and Automation}.\hskip 1em plus 0.5em minus 0.4em\relax IEEE, 2021, pp. 8111--8117.

\bibitem{pourghorban2022target}
A.~Pourghorban, M.~Dorothy, D.~Shishika, A.~Von~Moll, and D.~Maity, ``Target defense against sequentially arriving intruders,'' in \emph{61st Conference on Decision and Control}.\hskip 1em plus 0.5em minus 0.4em\relax IEEE, 2022, pp. 6594--6601.

\bibitem{fuchs2010cooperative}
Z.~E. Fuchs, P.~P. Khargonekar, and J.~Evers, ``Cooperative defense within a single-pursuer, two-evader pursuit evasion differential game,'' in \emph{49th IEEE conference on decision and control (CDC)}.\hskip 1em plus 0.5em minus 0.4em\relax IEEE, 2010, pp. 3091--3097.

\bibitem{yan2021cooperative}
R.~Yan, Z.~Shi, and Y.~Zhong, ``Cooperative strategies for two-evader-one-pursuer reach-avoid differential games,'' \emph{International Journal of Systems Science}, vol.~52, no.~9, pp. 1894--1912, 2021.

\bibitem{liang2020analysis}
L.~Liang, F.~Deng, M.~Lu, and J.~Chen, ``Analysis of role switch for cooperative target defense differential game,'' \emph{IEEE Transactions on Automatic Control}, vol.~66, no.~2, pp. 902--909, 2020.

\bibitem{von2019multi}
A.~Von~Moll, D.~Casbeer, E.~Garcia, D.~Milutinovi{\'c}, and M.~Pachter, ``The multi-pursuer single-evader game: A geometric approach,'' \emph{Journal of Intelligent \& Robotic Systems}, vol.~96, pp. 193--207, 2019.

\bibitem{von2020robust}
A.~Von~Moll, M.~Pachter, E.~Garcia, D.~Casbeer, and D.~Milutinovi{\'c}, ``Robust policies for a multiple-pursuer single-evader differential game,'' \emph{Dynamic Games and Applications}, vol.~10, pp. 202--221, 2020.

\end{thebibliography}

\appendix
\subsection{Proof of Lemma~\ref{lem:W_properties}} \label{AP:W_properties} 

1) One may verify that both $W_1, W_2 \succeq 0$.
In order to show that $W \succ 0$, we will use the method of contradiction as follows. 
Let us assume that there exists $\mathbf{v}\in \Re^N$ such that $\mathbf{v}\T W \mathbf{v} = 0$. 
Therefore, we must have $\mathbf{v}\T W_1 \mathbf{v} = \mathbf{v}\T W_2 \mathbf{v} = 0$.
Due to \Cref{assm:communication}, we may conclude that $\mathbf{v}\T W_1 \mathbf{v} = 0$ if and only if $\mathbf{v} = c \mathds{1}$ for some $c \in \Re$.
However, $ (c \mathds{1}) \T W_2 (c \mathds{1}) = c^2 \sum_{i=1}^N b_i$. 
Due to \Cref{assm:sensing}, we may conclude that $\sum_{i=1}^N b_i > 0$, and therefore, $ (c \mathds{1}) \T W_2 (c \mathds{1}) = 0$ if and only if $c = 0$. 
Consequently, $\mathbf{v}\T W \mathbf{v} = 0$ if and only if $\mathbf{v}$ is a zero vector. 
Thus, $W$ is a positive definite matrix.\\

2) Let $V_1 = \textrm{span}\{\mathds{1}\}$ and $V_2 = \textrm{ker}(\mathds{1}\T)$. Consequently $\Re^N = V_1 \oplus V_2$, and $V_1 \perp V_2$.  
Any $v\in \Re^N$ can be expressed as $v = \alpha v_1 + \beta v_2$ where $\alpha, \beta \in \Re$, $v_i \in V_i$ and $\|v_i\| =1$, for $i=1,2$. 
Consequently, $\|v\|^2 = \alpha^2 + \beta^2$, and $v_1 = \frac{1}{\sqrt{N}}\mathds{1}$.
Let $v$ be the eigenvector corresponding to the minimum eigenvalue of $W$. 
Therefore,
\begin{align} \label{eq:lambda_min}
   \lambda_{\min}(W) \|v\|^2 & = v\T W v = v\T W_1 v + v\T W_2 v \nonumber \\
   & = \beta^2 v_2\T W_1 v_2 + \sum_{i=1}^N b_i (\frac{\alpha}{\sqrt{N}} + \beta v_{2i})^2, 
\end{align}
where $v_{2i}$ is the $i$-th component of $v_2$. 
Simplifying \eqref{eq:lambda_min} yields
\begin{align} \label{eq:mineigenval}
    \lambda_{\min}(W) & = \frac{\beta^2 v_2\T W_1 v_2 + \sum_{i=1}^N b_i (\frac{\alpha}{\sqrt{N}} + \beta v_{2i})^2}{\alpha^2 + \beta^2} \nonumber \\
    & \ge \min_{\stackrel{\alpha, \beta, v_2}{\|v_2\|=1,~ v_2 \in V_2}}  \frac{\beta^2 v_2\T W_1 v_2 + \sum_{i=1}^N b_i (\frac{\alpha}{\sqrt{N}} + \beta v_{2i})^2}{\alpha^2 + \beta^2}.
\end{align}
Since $v_2 \in V_2$ and $\|v_2\|=1$, we may write $v_2\T W_1 v_2 \ge \lambda_2(W_1)$, where $\lambda_2(\cdot)$ denotes the second smallest eigenvalue of a matrix.
Now, let us separately consider the cases $\beta = 0$ and $\beta \ne 0$ for the minimization in \eqref{eq:mineigenval}.
$\beta = 0$ produces $\frac{m}{N}$ on the r.h.s. of \eqref{eq:mineigenval}, where $m = \sum_{i}^N b_i$. 
On the other hand, $\beta \ne 0$ results in 
\begin{align} \label{eq:f_gamma}
    \min_{\stackrel{\alpha, \beta \ne 0, v_2}{\|v_2\|=1,~ v_2 \in V_2}}& \frac{\lambda_2(W_1)  + \sum_{i=1}^N b_i (\frac{\alpha}{\beta\sqrt{N}} + v_{2i})^2}{\frac{\alpha^2}{\beta^2} + 1} \nonumber \\ 
    &= 
    \min_{\stackrel{\gamma, v_2}{\|v_2\|=1,~ v_2 \in V_2}}\frac{\lambda_2(W_1) + \sum_{i=1}^N b_i (\frac{\gamma}{\sqrt{N}} + v_{2i})^2}{\gamma^2 + 1} \nonumber \\
    & = \min_{\gamma} \frac{\lambda_2(W_1) + \frac{m}{N}\left(|\gamma| - \sqrt{\frac{N-m}{m}} \right)^2}{\gamma^2 + 1},
\end{align}
where the last equality is obtained by minimizing over $v_2$ and substituting the optimal $v_2^*$, where $v_{2i}^* = -\sign(\gamma) \sqrt{\frac{N-m}{Nm}}$ for $i$'s such that $b_i = 1$, and $v_{2i}^* = \sign(\gamma) \sqrt{\frac{m}{N(N-m)}}$ for $i$'s such that $b_i = 0$.
For brevity, let us define $f(\gamma)$ to be the function under minimization in \eqref{eq:f_gamma}.
Therefore, combining the cases $\beta=0$ and $\beta \ne 0$, we may write 
\begin{align*}
    \lambda_{\min}(W)  & \ge \min \left\{ \frac{m}{N}, \quad \min_{\gamma} f(\gamma)\right\}  = \min_{\gamma} f(\gamma),
\end{align*}
where the last equality is obtain by observing that $\min_\gamma f(\gamma) \le \lim_{\gamma \to \infty} f(\gamma) = m/N$.

\end{document}